\documentclass[11pt]{book}

\usepackage[dvips]{epsfig}

\usepackage{mathtools}
\usepackage{amssymb}
\usepackage{multirow}
\usepackage{amsbsy}
\usepackage{pstricks}
\usepackage{physics}
\usepackage{tikz-cd}

\newenvironment{proof}{\paragraph*{Proof:}}{\hfill$\square$}

\catcode`\@=11

\def\@normalsize{\@setsize\normalsize{13pt}\xipt\@xipt
	\abovedisplayskip 11pt plus3pt minus6pt
	\belowdisplayskip \abovedisplayskip
	\abovedisplayshortskip \z@ plus3pt
	\belowdisplayshortskip 6.6pt plus3.5pt minus3pt} 

\def\small{\@setsize\small{12pt}\xipt\@xipt
	\abovedisplayskip 10pt plus2pt minus5pt
	\belowdisplayskip \abovedisplayskip
	\abovedisplayshortskip \z@ plus3pt
	\belowdisplayshortskip 6pt plus3pt minus3pt
	\def\@listi{\topsep 6pt plus 2pt minus 2pt
		\parsep 3pt plus 2pt minus 1pt
		\itemsep \parsep}}

\def\footnotesize{\@setsize\footnotesize{10pt}\ixpt\@ixpt
	\abovedisplayskip 8pt plus 2pt minus 4pt
	\belowdisplayskip \abovedisplayskip
	\abovedisplayshortskip \z@ plus 1pt
	\belowdisplayshortskip 4pt plus 2pt minus 2pt
	\def\@listi{\topsep 4pt plus 2pt minus 2pt
		\parsep 2pt plus 1pt minus 1pt
		\itemsep \parsep}}

\def\scriptsize{\@setsize\scriptsize{9.5pt}\viiipt\@viiipt}
\def\tiny{\@setsize\tiny{7pt}\vipt\@vipt}
\def\large{\@setsize\large{14pt}\xiipt\@xiipt}
\def\Large{\@setsize\Large{18pt}\xivpt\@xivpt}
\def\LARGE{\@setsize\LARGE{22pt}\xviipt\@xviipt}
\def\huge{\@setsize\huge{25pt}\xxpt\@xxpt}
\def\Huge{\@setsize\Huge{30pt}\xxvpt\@xxvpt}

\def\section{\@startsection {section}{1}{\z@}%
	{-1.5\baselineskip plus-1pt minus-3pt}{1\baselineskip plus1pt minus2pt}%
	{\centering\normalsize\bf}}
\def\subsection{\@startsection{subsection}{2}{\z@}%
	{-1\baselineskip plus-1pt minus-2pt}{1\baselineskip plus1pt minus2pt}%
	{\normalsize\sc\noindent}}
\def\subsubsection{\@startsection{subsubsection}{3}{\z@}%
	{-1\baselineskip plus-1pt minus-2pt}{1sp}{\normalsize\it\noindent}}
\def\paragraph{\@startsection{paragraph}{4}{\z@}%
	{1\baselineskip plus1pt minus2pt}{-1em}{\normalsize\it\noindent}}
\let\subparagraph=\paragraph

\def\tableofcontents{\@restonecolfalse\if@twocolumn\@restonecoltrue
	\onecolumn\fi\OSIDcont\@starttoc{con}\if@restonecol\twocolumn\fi}

\def\l@section{\@dottedtocline{1}{0em}{.66em}}

\def\thebibliography#1{\section*{{Bibliography}\@mkboth
		{BIBLIOGRAPHY}{BIBLIOGRAPHY}}\footnotesize\rm\list
	{[\arabic{enumi}]}{\settowidth\labelwidth{[#1]}\leftmargin\labelwidth
		\advance\leftmargin\labelsep\usecounter{enumi}}
	\def\newblock{\hskip .11em plus .33em minus -.07em}
	\sloppy\clubpenalty4000\widowpenalty4000
	\sfcode`\.=1000\relax}

\def\ps@myheadings{\let\@mkboth\@gobbletwo
	\def\@oddhead{\hfil{\footnotesize\rm\rightmark}\hfil}
	\def\@evenhead{\hfil{\footnotesize\rm\leftmark}\hfil}
	\def\@oddfoot{\hfil{\footnotesize\sf\artid-\thepage}\hfil}
	\def\@evenfoot{\hfil{\footnotesize\sf\artid-\thepage}\hfil}
	\def\sectionmark##1{}\def\subsectionmark##1{}}

\newcounter{paPer}     %
\def\ps@osiD{\let\@mkboth\@gobbletwo
	\def\@oddfoot{\hfil{\footnotesize\sf\artid-\thepage}\hfil}
	\def\@evenhead{}\let\@evenfoot\@oddfoot}

\def\cite{\@ifnextchar [{\@tempswatrue\@Rcitex}{\@tempswafalse\@Rcitex[]}}

\def\@Rcitex[#1]#2{\if@filesw\immediate\write\@auxout{\string\citation{#2}}\fi
	\def\@citea{}\@cite{\@for\@citeb:=#2\do
		{\@citea\def\@citea{,\penalty\@m\,}\@ifundefined
			{b@\the\value{paPer}R\@citeb}{{\bf ?}\@warning
				{Citation `\@citeb' on page \thepage \space undefined}}%
			\hbox{\csname b@\the\value{paPer}R\@citeb\endcsname}}}{#1}}

\long\def\@caption#1[#2]#3{\par\addcontentsline{\csname
		ext@#1\endcsname}{#1}{\protect\numberline{\csname
			the#1\endcsname}{\ignorespaces #2}}\begingroup
	\@parboxrestore
	\small                                        
	\@makecaption{\csname fnum@#1\endcsname}{\ignorespaces #3}\par
	\endgroup}

\catcode`\@=12

\let\Rlabel=\label
\let\Rbibitem=\bibitem
\let\Rref=\ref
\let\Rpageref=\pageref
\def\label#1{\expandafter\Rlabel{\the\value{paPer}R#1}}
\def\bibitem#1{\expandafter\Rbibitem{\the\value{paPer}R#1}}
\def\ref#1{\expandafter\Rref{\the\value{paPer}R#1}}
\def\pageref#1{\expandafter\Rpageref{\the\value{paPer}R#1}}

\def\thesection{\arabic{section}.}

\def\YYMm{\rule{0ex}{4em}}
\newtoks\TITsi
\newtoks\TITsii

\def\title#1{\def\TITs{\LARGE{\raggedright\noindent\YYMm #1%
			\vskip8pt\par}}}

\def\author#1{\autMM{#1}\def\LHD{#1}}
\def\and{{\rm\lowercase{and}}}

\def\autMM#1{\TITsii={\vskip10pt\par\normalsize\rm\noindent #1\par}%
\TITsi=\expandafter{\TITs}\edef\TITs{\the\TITsi\the\TITsii}}

\def\address#1{\TITsii={\vskip6pt\par\footnotesize\sl\noindent #1\par}%
\TITsi=\expandafter{\TITs}%
\edef\TITs{\the\TITsi\the\TITsii}}

\def\received#1{\TITsii={\vskip10pt\par\small\rm\noindent(Received: #1)\par}%
\TITsi=\expandafter{\TITs}\edef\TITs{\the\TITsi\the\TITsii}}

\def\headtitle#1{\def\RHD{#1}}
\def\headauthor#1{\def\LHD{#1}}

\def\abst{{\bf Abstract.}}
\def\abstract#1{\TITs
\vskip15pt\par\noindent
{\footnotesize{\abst~} #1\vskip3pt\par}
\markright{\RHD}
\markboth{\LHD}{\RHD}}

\def\startpaper{%
\cleardoublepage
\setcounter{section}{0}
\stepcounter{paPer}
\setcounter{equation}{0}
\setcounter{footnote}{0}
\setcounter{figure}{0}
\setcounter{table}{0}
\def\theequation{\arabic{equation}}
\def\thefootnote{\arabic{footnote}}
\setcounter{defn}{0}
\setcounter{thm}{0}
\setcounter{lem}{0}
\setcounter{prop}{0}
\setcounter{rem}{0}
\thispagestyle{osiD}}

\def\OSIDcont{\cleardoublepage\thispagestyle{empty}
\markright{}\markboth{}{}
\normalsize\rm
\hspace*{\fill}{\large\rm
Contents of the Volume \Volume, Number \Number}\hspace*{\fill}
\par\vspace{1.5em}
\par\noindent}

\def\endpaper{\expandafter\label{\the\value{paPer}OpSy}}

\def\1{{\mathchoice{\rm 1\mskip-4mu l}{\rm 1\mskip-4mu l}%
{\rm 1\mskip-4.5mu l}{\rm 1\mskip-5mu l}}}

\def\varkappa{\mbox{\bBB\char 123}}

\def\longhookrightarrow{\lhook\joinrel\relbar\joinrel\rightarrow}

\def\longhookUp{\lower6pt\hbox{\rotatebox{90}{$\longhookrightarrow$}}}

\def\Tr{\mathop{\rm Tr}}

\newtheorem{thm}{\rm THEOREM}

\newtheorem{prop}{\rm PROPOSITION}

\newtheorem{defn}{\rm DEFINITION}
\newtheorem{rem}{\it Remark}

\def\theequation{\thesection\arabic{equation}}

\def\Myskip{\setlength{\baselineskip}{13pt}}

\def\text#1{\quad\mbox{\rm  #1 }\quad}

\def\DOInumber{}

\input xy
\xyoption{all}

\InputIfFileExists{psfig.sty}{\typeout{^^Jpsfig.sty inputed...ok}}{\typeout{^^JWarning: psfig.sty could not be found.^^J}}

\usepackage{amsmath}
\usepackage{hyperref}

\begin{document}
	\emergencystretch 3em

\def\artid{0000001}
\def\Volume{15}
\def\Number{1}
\def\Year{2008}
\setcounter{page}{1}

\def\DOInumber{}

\startpaper

\newcommand{\Ker}{\mathrm{Ker}}
\newcommand{\Mn}{M_n(\mathbb{C})}
\newcommand{\Mk}{M_k(\mathbb{C})}
\newcommand{\id}{\mbox{id}}
\newcommand{\ot}{{\,\otimes\,}}
\newcommand{{\Cd}}{{\mathbb{C}^d}}
\newcommand{\sbsigma}{{\mbox{\scriptsize \boldmath $\sigma$}}}
\newcommand{\sbalpha}{{\mbox{\scriptsize \boldmath $\alpha$}}}
\newcommand{\sbbeta}{{\mbox{\scriptsize \boldmath $\beta$}}}
\newcommand{\bsigma}{{\mbox{\boldmath $\sigma$}}}
\newcommand{\balpha}{{\mbox{\boldmath $\alpha$}}}
\newcommand{\bbeta}{{\mbox{\boldmath $\beta$}}}
\newcommand{\bmu}{{\mbox{\boldmath $\mu$}}}
\newcommand{\bnu}{{\mbox{\boldmath $\nu$}}}
\newcommand{\ba}{{\mbox{\boldmath $a$}}}
\newcommand{\bb}{{\mbox{\boldmath $b$}}}
\newcommand{\sba}{{\mbox{\scriptsize \boldmath $a$}}}
\newcommand{\MD}{\mathfrak{D}}
\newcommand{\sbb}{{\mbox{\scriptsize \boldmath $b$}}}
\newcommand{\sbmu}{{\mbox{\scriptsize \boldmath $\mu$}}}
\newcommand{\sbnu}{{\mbox{\scriptsize \boldmath $\nu$}}}
\def\oper{{\mathchoice{\rm 1\mskip-4mu l}{\rm 1\mskip-4mu l}%
		{\rm 1\mskip-4.5mu l}{\rm 1\mskip-5mu l}}}
\def\<{\langle}
\def\>{\rangle}
\def\theequation{\thesection\arabic{equation}}
\newcommand\sbullet[1][.64]{\mathbin{\vcenter{\hbox{\scalebox{#1}{$\bullet$}}}}}

\thispagestyle{plain}

\title{Bridging Classical and Quantum Worlds: Maps, States, and Evolutions}
\author{Daniele Amato$^{1,2,*}$, Paolo Facchi$^{1,2}$, and Giuseppe Marmo$^{3,4}$}
\address{$^1$Dipartimento di Fisica, Universit\`a di Bari, I-70126 Bari, Italy} 
\address{$^2$INFN, Sezione di Bari, I-70126 Bari, Italy}
\address{$^3$Dipartimento di Fisica “E. Pancini”, Universit{\`a} di Napoli Federico II, Naples, Italy}
\address{$^4$INFN-Sezione di Napoli, Naples, Italy}
\address{$^*$Corresponding author: \href{mailto:daniele.amato@ba.infn.it}
{\texttt{daniele.amato@ba.infn.it}}}
\headauthor{D. Amato, P. Facchi, and G. Marmo}
\headtitle{Bridging Classical and Quantum Worlds}
\received{\today}

\abstract{In this work, we present several aspects of the interplay between classical and quantum theories. After reviewing the equivalence between positivity and complete positivity in the commutative setting, we introduce and analyze intermediate notions that interpolate between these two properties for linear maps on the space of operators on a Hilbert space, highlighting their natural connections to quantum dynamics and entanglement theory. Finally, we discuss possible constructions that allow one to derive quantum states and evolutions from classical counterparts, and vice versa.}

\Myskip


\section{Introduction}

\label{intro_sec}
Classical and quantum theories exhibit various differences both formally and conceptually. 
One of the main features of quantum theory is that the algebra of observables is non-commutative, implying the possible incompatibility between the measurements on a quantum system and, from a dynamical perspective, the need of complete positivity in the description of quantum dynamics. 

Positive and completely positive maps play a central role in the theory of C$^\ast$-algebras, with profound implications in both pure mathematics and quantum physics. A positive map between C$^\ast$-algebras is a linear map that sends positive elements into positive elements, preserving the order structure of the algebra. Although positivity is a natural and minimal requirement, it is often insufficient to guarantee desirable mathematical or physical properties. 
This motivates the study of the special subclass of completely positive maps, for which positivity holds even if we extend the map to matrices of arbitrary order over the algebra by tensoring it with the identity map, cf. Definition~\ref{CP_def}. Importantly, the set of completely positive maps involves states and conditional expectations, which are the building blocks of quantum probability~\cite{accardi1981topics}, as well as representations in terms of which one can express any completely positive map, up to enlarging the algebra, via Stinespring's dilation theorem~\cite{stinespring1955positive}. The latter result is one of the various characterizations of complete positivity which cannot be easily extended to positive maps, whose structure is still unclear even in the finite-dimensional case. However, if either the initial $C^\ast$-algebra or the target one is commutative, then it is well known that positivity and complete positivity are equivalent notions~\cite{arveson1969,stinespring1955positive}. In particular, this means that there is no need to introduce complete positivity for classical dynamical systems. Instead, complete positivity is stronger than positivity if both the initial and target $C^\ast$-algebras are noncommutative, as illustrated by the transposition map~\cite{arveson1969,tomiyama1983transpose}, see also~\cite{Choi_1972,stinespring1955positive} for other examples.

The connection between classical and quantum frameworks can be also studied from a different perspective, which is the so-called quantum-to-classical transition, namely, the emergence of classically observable features in the microscopic quantum world~\cite{schlosshauer2007decoherence}. The latter issue can be attacked in a variety of ways, for instance an approach promoted by De Broglie and Bohm~\cite{bohm1952suggested,bohm1952suggested2,de1927mecanique} and further developed in the next years~\cite{haba1998classical,hepp1974classical} consists of obtaining the classical description from the quantum one in the semiclassical limit $\hbar \mapsto 0$. Alternatively, the quantum-to-classical transition problem can be also addressed from a dynamical point of view~\cite{gell1993classical}, which may be used to model e.g.\ a measurement~\cite{jauch1964problem,von2018mathematical}. In particular, classical behaviour of an open quantum system may emerge in the large-time limit, once dissipation and decoherence are complete~\cite{alipour2017dynamical,amato2023asymptotics,Amato2025,AFK_OSID,blanchard1993interaction,chruscinski2012observables}. Finally, observe that seminal ideas such as Koopman's formalism~\cite{koopman1931}, phase-space quantization~\cite{GROENEWOLD1946,moyal1949quantum,weyl1927quantenmechanik,wigner1932quantum}, and Wentzel–Kramers–Brillouin (WKB) approximation~\cite{brillouin1926mecanique,kramers1926wellenmechanik,wentzel1926verallgemeinerung}, 
as well as more recent works~\cite{man1997spin,man2008semigroup} are based on the use of standard classical (quantum) tools in the quantum (classical) realm. This idea passes through, for example, the construction of a quantum (classical) state or dynamics from a classical (quantum) one, which will be explored in some depth in Section~\ref{class_quant_sec}.

In this paper, we shall overview several aspects of the interplay between classical and quantum theories. Specifically, the work is organized as follows: classical and quantum Markov chains, ruling discrete-time classical and quantum stochastic evolutions, will be recalled in  Section~\ref{Markov_sec}, while  positivity, complete positivity, and various properties lying between them will be discussed in Section~\ref{pos_compl_pos_sec}, together with their relation with entanglement theory. Afterwards,  we shall present in Section~\ref{class_quant_sec} how to embed a classical dynamics into a quantum one and, conversely, how to reduce quantum evolutions into classical ones. Finally, we draw the conclusions of the work in Section~\ref{concl_sec}.


\section{Markov chains}
\label{Markov_sec}
In this section we will overview the role that (complete) positivity plays in the description of stochastic evolutions.  The next subsection covers the classical case, while quantum stochastic dynamics will be discussed in Subsection~\ref{open_dyn_sec}.

\subsection{Classical case}
\label{classic_sec}
 The goal of the present subsection is to describe a classical discrete-time stochastic evolution on a finite space of events or states $\Omega = \{ 1, \dots , d \}$. First, \textit{observables} are given by real vectors $\vec{x} \in \mathbb{R}^d$, forming the Hermitian part of the commutative $C^\ast$-algebra $\mathcal{A} = (\mathbb{C}^d , \sbullet , \ast , \| \cdot \|)$, where 
\begin{align}
 \vec{x} \sbullet  \vec{y} &:= (x_1 y_1, x_2 y_2, \dots, x_d y_d)^T =(x_i y_i)_{i=1}^d, \\
\vec{x}^{\,\ast} &:= (x_{i}^\ast)_{i=1}^d, \qquad \| \vec{x} \| :=  \max_{i \in \{ 1, \dots , d \} } |x_i|.
\end{align}
In particular, any positive observable $\vec{x}$, or shortly $\vec{x} \geqslant 0$, is of the form $\vec{x} = \vec{y}^{\,\ast} \sbullet \vec{y}$ for some $\vec{y} \in \mathbb{C}^d$.
Moreover, a \textit{state} $\varphi$, defined as a linear functional $\varphi : \mathcal{A} \mapsto \mathbb{C}$ on $\mathcal{A}$ satisfying the properties
\begin{equation}
\vec{x} \geqslant 0 \Rightarrow \varphi(\vec{x}) \geqslant 0, \quad \varphi(\vec{e}) = 1,  
\end{equation}
with $\vec{e} = (1, \dots , 1)^T$ being the identity vector of $\mathbb{C}^d$,
is uniquely associated with a probability vector $\vec{p} = (p_i)_{i=1}^d$, with $\vec{p}\geqslant 0$ and $\sum_{i=1}^d p_i=1$, via 
\begin{equation}
  \varphi(\vec{x}) = \vec{p} \cdot  \vec{x} := \sum_{i=1}^d p_i x_i  , \quad \vec{x} \in \mathbb{C}^d,
\end{equation}
where $\vec{x} \cdot \vec{y}$ denotes the Euclidean scalar product of two vectors $\vec{x}, \vec{y} \in \mathbb{C}^d$. Therefore, without possibility of confusion, we will often refer to a probability vector $\vec{p} = (p_i)_{i=1}^d$ as a state.  
 Physically, if $ \vec{x} \in \mathbb{R}^d$ is an observable and $\vec{p}$ a state, then we can interpret the $i$-th component $p_i$ of $\vec{p}$ as the probability to obtain after a measurement the $i$-th possible value $x_i$ of $\vec{x}$ and, consequently, $\varphi(\vec{x})$ as the expectation value of $\vec{x}$ over the state $\vec{p}$. 
 
 That being said, under the homogeneity and Markovian assumptions, a classical stochastic dynamics on $\Omega$ is described by a \textit{classical Markov chain} on $\Omega$, totally encoded into a \textit{row-stochastic matrix} $S = (S_{ij})_{i , j=1}^d$ defined by~\cite{bremaud2013markov}
\begin{align}
S_{ij} \geqslant 0, \quad i,j=1, \dots , d, \label{pos_cond}\\
\sum_{j=1}^d S_{i j} = 1, \quad i=1,\dots , d. \label{norm_row-stoch}
\end{align}
The matrix $S$ describes the unit-time evolution of observables, $\vec{x}\mapsto S \vec{x}$, with the standard row-by-column matrix multiplication being understood, and it is often called the transition matrix of the Markov chain.
Moreover, as a consequence of the homogeneity assumption, if $\vec{x}_0$ is the observable at time $t=0$, then the evolved one $\vec{x}_n$ at time $t=n \geqslant 1$ reads
\begin{equation}
\vec{x}_0 \in \mathbb{R}^d\mapsto \vec{x}_n= S^n \vec{x}_0 \in \mathbb{R}^d,
\end{equation}
so the classical Markov chain reduces to the discrete-time semigroup $(S^n)_{n \in \mathbb{N}}$ induced by $S$ via the composition operation $(ST)\vec{x} := S(T\vec{x})$ with $S,T \in B(\mathbb{C}^d)$ and $\vec{x} \in \mathbb{C}^d$. Notice that the first property~\eqref{pos_cond} is known as (entry-wise) \textit{positivity}, while the second property~\eqref{norm_row-stoch} defining $S$ is equivalent to the \textit{unitality} condition $S \vec{e} = \vec{e}$, namely, the unit (trivial) observable $\vec{e}$ is a fixed point of the evolution. Moreover, the positivity property~\eqref{pos_cond} implies that $S_{ij} \in \mathbb{R}$ for all $i,j = 1, \dots , d$, which guarantee that $S \mathbb{R}^d \subseteq \mathbb{R}^d$, i.e., $S$ maps observables into observables. 

In order to understand even better the properties~\eqref{pos_cond}--\eqref{norm_row-stoch} of the transition matrix $S$, let us take into account the adjoint dynamics, involving the states and defined via the equality of the unit-time evolved expectation values 
\begin{equation}
\label{dual_cond}
\vec{p} \cdot (S\vec{x}) = (Q\vec{p}) \cdot \vec{x}, \quad \forall \vec{p} \in \Delta_d, \; \forall \vec{x} \in \mathbb{R}^d,
\end{equation}
with 
\begin{equation}
    \Delta_d := \left\{ \vec{p}\in \mathbb{R}^d \;:\; \vec{p}\geq 0,\, \sum_{i=1}^d p_i =1 \right\}
\end{equation}
denoting the probability simplex in $\mathbb{R}^d$.
This implies that $Q = S^T$, the transpose of $S$, or equivalently $Q = (Q_{ij})_{i , j =1}^d$ is a \textit{column-stochastic} matrix, namely, 
\begin{align}
Q_{ij} \geqslant 0, \quad i,j=1, \dots , d,\\
\sum_{i=1}^d Q_{ij} = 1, \quad j=1,\dots , d. \label{norm_col-stoch}
\end{align}
The quantity $Q_{ij}$ represents the unit-time probability of the transition $j \mapsto i$, and  property~\eqref{norm_col-stoch} amounts to say that the the total probability of transitions from state $j$ to all possible states in $\Omega$ sums to one. Put differently, the two properties are equivalent to impose that $Q \Delta_d \subseteq \Delta_d$, viz. $Q$ maps states into states, clarifying the defining properties of $S = Q^T$.   
Again, since the evolution is homogeneous, then the probability vector $\vec{p}_{0}$ at time $t=0$, i.e., the initial state of the chain, evolves over the interval $[0,n]$ as follows
\begin{equation}
\label{dyn_HMC}
\vec{p}_{0} \in \Delta_d \mapsto \vec{p}_{n} = Q^n \vec{p}_{0} \in \Delta_d,
\end{equation}
or, in other words, the dynamics at all (discrete) times is described by the semigroup $(Q^n)_{n \in \mathbb{N}}$ associated to $Q$ via composition. 
It is worth noting that the set of row-stochastic (respectively, column-stochastic) matrices forms a semigroup; that is, $QP$ is row-stochastic (respectively, column-stochastic) whenever $Q,P$ are. Therefore, for any row-stochastic matrix $Q$, $Q^n$ ($(Q^{T})^n$) is row-stochastic (column-stochastic) for any $n \geqslant 1$. Consequently, the set of matrices which are \textit{bistochastic}, that is, column- and row-stochastic, forms also a semigroup. 

In the continuous-time setting, a classical Markovian evolution of probability vectors  is described by a continuous semigroup  $(Q(t))_{t\in\mathbb{R}^+}$ of column-stochastic matrices or, equivalently, by the master equation
\begin{equation}
    \frac{d}{dt} \vec{p}(t) = L \vec{p}(t) \Rightarrow \vec{p}(t) = Q(t) \vec{p}(0) = e^{tL} \vec{p}(0) \in \Delta_d,
\end{equation}
Here, $L = (L_{ij})_{i,j=1}^d$, known as the \textit{Kolmogorov generator}, 
has the explicit structure
\begin{equation}
\label{Kolmogorv_struc}
    L_{ij} = W_{ij} - \delta_{ij} W_j, \quad W_j = \sum_{i=1}^d W_{ij}, \quad i,j=1, \dots , d,
\end{equation}
where $W_{ij} \geqslant 0$ with $i \neq j$ represents the transition rate from state $j$ to state $i$. As in the discrete-time setting, the classical Markovian dynamics of observables is obtained from the one of states by the duality condition~\eqref{dual_cond}.
\subsection{Quantum case}
\label{open_dyn_sec}
In this section we will discuss the dynamics of an \textit{open quantum system} or, borrowing the words of the seminal paper~\cite{sudarshan1961stochastic}, the stochastic dynamics of a quantum system, in a sense which will be clearer soon. 

In the standard description of quantum mechanics, a quantum system is described by a Hilbert space $\mathcal{H}$ which, if not stated otherwise, we will assume to be \emph{finite-dimensional} with dimension $d$. Moreover, observables are defined as the self-adjoint elements of the non-commutative $C^\ast$-algebra $(B(\mathcal{H}) , \sbullet , \ast , \| \cdot \|)$, where $B(\mathcal{H})$ denotes the space of bounded operators over $\mathcal{H}$, $\sbullet$ the composition product, $\ast$ the Hermitian conjugate operation, and $\| \cdot \|$ the operator norm in $B(\mathcal{H})$. The symbol $\sbullet$ to denote the composition of operators will be henceforth dropped for the sake of simplicity of notation.
Any positive semi-definite operator in $B(\mathcal{H})$, shortly denoted as $X \geqslant 0$, is of the form $X = Y^\ast Y$ for some $Y \in B(\mathcal{H})$.
Recall also that a state $\varphi$ is defined as a positive normalized linear functional on $B(\mathcal{H})$, namely,
\begin{equation}
    X \geqslant 0 \Rightarrow \varphi(X) \geqslant 0, \quad \varphi(\mathbb{I}) = 1, 
\end{equation}
where $\mathbb{I}$ stands for the identity operator of $B(\mathcal{H})$.
Moreover, any state
may be identified with a \textit{density operator} $\rho$, which is a positive semi-definite operator on $\mathcal{H}$ with unit trace, through the pairing 
\begin{equation}
\label{exp_val}
   \varphi(X) = \Tr(\rho X) = \braket{\rho}{X}_2, \quad X \in B(\mathcal{H}),
\end{equation}
induced by the \textit{Hilbert-Schmidt inner product} $\braket{X}{Y}_2 := \Tr(X^\ast Y)$ with $X,Y \in B(\mathcal{H})$.
In particular,~\eqref{exp_val} expresses the expectation value of the observable $X = X^\ast \in B(\mathcal{H})$ over the state $\rho$.
Let $\mathcal{S}(\mathcal{H})$ be the convex and compact set of density operators on $\mathcal{H}$, which can be always expressed as 
\begin{equation}
\label{ens_dec}
   \!\! \!\!\!\rho = \sum_{i=1}^N p_i \ketbra{\psi_i}{\psi_i}, \quad \vec{p} = (p_i)_{i=1}^N \in \Delta_N, \,\ket{\psi_{i}} \in \mathcal{H}, \| \psi_i \|=1 , 
\end{equation}
i.e., as an ensemble of \textit{pure}, that is idempotent, density operators weighted by some probability distribution. In this sense, an open-system evolution may be regarded as stochastic, as it involves not only pure states, $\ketbra{\psi}{\psi}$, as in the closed-system case, but also their statistical mixtures. 

Now, we would like to describe the homogeneous discrete-time evolution of an open quantum system in the \textit{Heisenberg picture}, where we look how observables evolve in time. First, since, in the quantum realm, observables are Hermitian operators over $\mathcal{H}$, their unit-time evolution can be reasonably described by a linear map $\Phi$ on $B(\mathcal{H})$, see~\cite{gisin1990weinberg} for a discussion of the linearity assumption. Secondly, the homogeneity assumption necessarily implies that the dynamics at all times reduces to the discrete-time semigroup $(\Phi^n)_{n \in \mathbb{N}}$ associated to $\Phi$ via composition, where $\Phi^n = \Phi \circ \dots \circ \Phi$ denotes the $n$-th fold composition of $\Phi$. Thus it remains to understand the properties characterizing the unit-time evolution $\Phi$. By analogy with the classical case, we can assume that 
\begin{equation}
\Phi(X) = \Phi(X)^\ast, \;\, \forall X = X^\ast \in B(\mathcal{H}), \qquad \Phi(\mathbb{I}) = \mathbb{I},
\end{equation}
namely, the \textit{Hermiticity-preservation} and \textit{unitality conditions}, respectively. The first property ensures that observables are mapped into observables, while the second one amounts to impose that the (trivial) identity observable $\mathbb{I}$ is a fixed point of the evolution. However, as also suggested by the classical case, we shall see that the property of Hermiticity-preservation is not sufficient for a legitimate description of the evolution and, to this purpose, it is again worthwhile to switch to the adjoint dynamics involving the density operators, known as the \textit{Schr{\"o}dinger picture} of the system's dynamics. 

The equivalence of the two pictures is established through Dirac’s
prescription~\cite{Dirac1958}
\begin{equation}
\label{dual_cond_quant}
\braket{\Psi(\rho)}{X}_2 = \braket{\rho}{\Phi(X)}_2, \quad \forall X = X^\ast \in B(\mathcal{H}),\; \forall \rho \in \mathcal{S}(\mathcal{H}),
\end{equation}
of equality between the evolved expectation values in the Schr{\"o}dinger and Heisenberg pictures, thus we infer that the unit-time Schr{\"o}dinger evolution $\Psi$ is the Hilbert-Schmidt adjoint $\Phi^\dagger$ of $\Phi$. For sure, in the Schr{\"o}dinger picture, the unit-time dynamics must map states into states, namely, $\Phi^\dagger(\mathcal{S}(\mathcal{H})) \subseteq \mathcal{S}(\mathcal{H})$, meaning that 
\begin{equation}
\label{PTP}
X \geqslant 0 \Rightarrow \Phi^\dagger(X) \geqslant 0, \qquad \Tr(\Phi^\dagger(X)) = \Tr(X), \quad \forall X \in B(\mathcal{H}).
\end{equation}
 The first condition in~\eqref{PTP}, called \textit{positivity}(-preservation), is the quantum analogue of the first property in~\eqref{pos_cond}, while the second one, known as \textit{trace-preservation}, is the quantum counterpart of~\eqref{norm_row-stoch}. So any Schr{\"o}dinger dynamics must be a positive trace-preserving map on $B(\mathcal{H})$. 
 
 However, we will now show that positivity is not enough for a suitable description of the evolution of open quantum systems. Indeed, given any bipartite quantum system $A + B$ described by a Hilbert space $\mathcal{H} = \mathcal{H}_{A} \otimes \mathcal{H}_{B}$, if $\Phi_{A}$ and $\Phi_{B}$ are positive trace-preserving maps on $B(\mathcal{H}_{A})$ and $B(\mathcal{H}_{B})$, respectively, then the map
\begin{equation}
\Phi_A \otimes \Phi_B : B(\mathcal{H}_{A} \otimes \mathcal{H}_{B}) \rightarrow B(\mathcal{H}_{A} \otimes \mathcal{H}_{B})
\end{equation} 
is a trace-preserving, but in general \textit{not} positive. In other words, positivity is not stable under tensor products, a fact exemplified by the transposition map~\cite{arveson1969}. Physically, this means that, even if the evolutions $\Phi_A$ and $\Phi_B$ of the two subsystems are legitimate, this may not be true for the dynamics of the composite system $A+B$. More precisely, while any \textit{separable} state $\rho_{AB}$, that is, of the form 
\begin{equation}
\label{sep_def}
\rho_{AB}= \sum_{i=1}^N p_i \rho_{i}^{(A)} \otimes \rho_{i}^{(B)},
\end{equation}
with $\vec{p} = (p_i)_{i=1}^d \in \Delta_d$ and $\rho_{i}^{(A)}$ ($\rho_{i}^{(B)}$) states on $\mathcal{H}_{A}$ ($\mathcal{H}_{B}$), is mapped into a (separable) state by $\Phi_A \otimes \Phi_B$, this is not generally true for non-separable, \textit{entangled} states. So we need to introduce a property guaranteeing that the evolution also maps entangled density operators of $A + B$ into density operators. To this purpose, let us introduce the $n$-\textit{fold ampliation} of a positive map $\Phi$
\begin{equation}
\label{ampl_map}
\Phi^{(n)} :=  \mathrm{id}_n \otimes \Phi : \mathcal{M}_{n}(\mathbb{C}) \otimes B(\mathcal{H})\rightarrow   \mathcal{M}_{n}(\mathbb{C}) \otimes B(\mathcal{H}),
\end{equation} 
where $\mathrm{id}_n$ is the identity operator on $B(\mathbb{C}^n) \simeq \mathcal{M}_{n}(\mathbb{C})$, the space of complex matrices of order $n$. If the $n$-th ampliation map $\Phi^{(n)}$ is positive for any $n \geqslant 1$, then $\Phi$ is called \textit{completely positive}.    

That being said, if $\Phi_A$ and $\Phi_B$ are completely positive trace-preserving maps on $B(\mathcal{H}_A)$ and $B(\mathcal{H}_{B})$ with $d_i := \dim(\mathcal{H}_i)$ and $i=A , B$, then
\begin{equation}
\Phi_A \otimes \Phi_B = (\Phi_A \otimes \mathrm{id}_{d_B})(\mathrm{id}_{d_A} \otimes \Phi_B)
\end{equation}
is a positive trace-preserving map, namely, it maps states into states, because composition of such maps. Therefore, we conclude that the Schr{\"o}dinger dynamics of an open quantum system with Hilbert space $\mathcal{H}$ is described by a completely positive trace-preserving  map $\Phi^\dagger$ on $B(\mathcal{H})$, called a \textit{quantum channel} in the physics literature. Coming back to the Heisenberg picture, described in the unit-time by $\Phi$, one can readily check that the unitality condition of $\Phi$ is equivalent to trace preservation of the Schr{\"o}dinger evolution $\Phi^\dagger$ and, moreover, complete positivity of $\Phi^\dagger$ is equivalent to the one of $\Phi$~\cite{heinosaari2011mathematical}. 

Thus, the Heisenberg unit-time dynamics is given by a Unital Completely Positive (\textit{UCP}) map $\Phi$ on $B(\mathcal{H})$. The semigroup $(\Phi^n)_{n \in \mathbb{N}}$ associated with the UCP map $\Phi$ gives the evolution at all discrete times, and it may be called \textit{quantum Markov chain}, following the jargon of classical probability theory. 
Analogously, the Schr{\"o}dinger dynamics at all discrete times is described by the semigroup $((\Phi^\dagger)^n)_{n\in\mathbb{N}}$ induced by the quantum channel $\Phi^\dagger$. 
Note that, consistently, $\Phi^n$ [respectively, $(\Phi^\dagger)^n$] is a UCP map (respectively, quantum channel), i.e., a valid Heisenberg (respectively, Schr{\"o}dinger) dynamics, for any $n \geqslant 1$ whenever $\Phi$ is. This is a consequence of the fact that the set of UCP maps (respectively, quantum channels) forms a semigroup, namely, if $\Phi$ and $\Psi$ are UCP maps (respectively, quantum channels), then so is $\Phi \circ \Psi$. Therefore, the set of unital and trace-preserving completely positive maps, or unital channels, forms a semigroup.

Similarly to the classical case, the continuous-time quantum Markovian evolution of an open quantum system is described by a continuous semigroup of quantum channels $(\Phi_t)_{t\in\mathbb{R}^+}$ with $\Phi_0 = \mathrm{id}_d$, the identity on $B(\mathcal{H})$, or, in an alternative way, by the Markovian quantum master equation
\begin{equation}
    \frac{d}{dt} \rho(t) = \mathcal{L} (\rho(t)) \Rightarrow \rho(t) = \Phi_t (\rho(0)) = e^{t\mathcal{L}} (\rho(0)) , \quad t \in \mathbb{R}^+.
\end{equation}
Recall that the generator $\mathcal{L}$ can be written in the Gorini-Kossakowski-Sudarshan-Lindblad (\textit{GKLS}) form~\cite{GKS76,Lindblad76}
\begin{equation}
\label{diag_GKLS}
\mathcal{L}(X)=-\mathrm i[H,X] + \sum_{k=1}^{N} \left(L_k X L_k^\ast - \frac{1}{2} \{ L_k^\ast L_k , X \} \right), \quad X \in B(\mathcal{H}),
\end{equation}
where $H=H^*$ is a Hermitian operator which may be interpreted as an effective Hamiltonian of the system, while the operators $\{L_k \}_{k=1}^N$ are known as \textit{noise} or \textit{jump operators}.
Note that the representation~\eqref{diag_GKLS} of the GKLS generators is the quantum analogue of the structure~\eqref{Kolmogorv_struc} of the Kolmogorov generators. Finally, notice that the continuous-time Markovian dynamics in the Heisenberg picture follows
from the one in the Schr{\"o}dinger scheme via Dirac's pairing~\eqref{dual_cond_quant}.
\section{Between positivity and complete positivity}
\label{pos_compl_pos_sec}
In the previous section, we introduced the notions of positivity and complete positivity in order to describe the stochastic evolution of classical and quantum systems. In the present section, we will discuss more in depth these properties and present some of the notions lying between positivity and complete positivity, introduced in order to analyze the intricate structure of the set of positive maps from different perspectives. Historically, interest in positive but not completely positive maps has already emerged in the fifties within the mathematical community since the seminal papers~\cite{kadison1952generalized,osaka1991positive,stinespring1955positive}, and it remains vivid over the next years~\cite{choi1974schwarz,choi1980some,robertson1983schwarz,woronowicz1976positive}, with a particular focus on idempotent positive maps~\cite{effros1979positive,hamana1979injective,osaka1991positive,stormer1982positive}. In fact, the topic has continued to receive a lot of attention in recent years~\cite{bhatia2000more,carlen2023characterizing,chruscinski2024class,mlynik2025characterization}, also in view of its connection with quantum dynamics, see Subsection~\ref{open_dyn_sec}, entanglement theory~\cite{horodecki1996necessary,peres1996separability}, as explained in the present section, and quantum information geometry~\cite{petz1996monotone,petz1996geometries}. For instance, several authors recently started to address the spectral consequences of (complete) positivity aimed at finding some non-trivial constraints for the relaxation of quantum systems towards equilibrium~\cite{chruscinski2021universal,chruscinski2025universal,muratore2025universal,vom2025universal}, see also~\cite{uhl2019affinity,xu2025thermodynamic} for similar efforts in the classical realm. 

Specifically, in the present section, we shall discuss the natural intermediate property of $n$-positivity, as well as the validity of Cauchy-Schwarz-like operator inequalities~\cite{kadison1952generalized,stinespring1955positive,woronowicz1976positive}, and we shall also mention the notion of $n$-copositivity~\cite{woronowicz1976positive}, involving the transposition map. 

We begin by recalling the notion of $n$-positivity, which we discuss --- together with positivity and complete positivity --- in the framework of $C^\ast$-algebras for later convenience. Recall that a positive element $a$ of a $C^\ast$-algebra $\mathcal{A}$, denoted with $a \geqslant 0$, is of the form $a = b^\ast b$ for some $b  \in \mathcal{A}$, as already noted when $\mathcal{A} = \mathbb{C}^n$ and $\mathcal{A} = B(\mathcal{H})$ in Subsections~\ref{classic_sec} and~\ref{open_dyn_sec}, respectively.  
\begin{defn}[$n$-positivity and complete positivity]
\label{CP_def}
Let $\Phi : \mathcal{A} \mapsto \mathcal{B}$ be a map between  $C^\ast$-algebras $\mathcal{A}$ and $\mathcal{B}$. Then  
\begin{itemize}
\item[]$i.$ $\Phi$ is called positive if and only if $\Phi(a) \geqslant 0$ for all $a \geqslant 0$;
\item[]$ii.$ A positive map $\Phi$ is said to be $n$-positive, for $n\geqslant 1$, if the amplified map $\Phi^{(n)}:= \mathrm{id}_n \otimes \Phi$ on $\mathcal{M}_{n}(\mathbb{C}) \otimes \mathcal{A}$, with $\mathcal{M}_{n}(\mathbb{C})$ being the algebra of complex matrices of order $n$ and $\mathrm{id}_n$ the identity, is positive;
\item[]$iii.$ $\Phi$ is said to be completely positive if $\Phi$ is $n$-positive for all $n \geqslant 1$.
\end{itemize}
\end{defn}
\begin{rem}
Observe that completely positive maps introduced in Definition~\ref{CP_def} should not be confused with completely positive matrices~\cite{berman2003completely}, namely, real positive semi-definite matrices $A = B^TB$ with $B$ being a positive matrix. 
\end{rem}
Clearly, a one-positive map is a positive map. Moreover, it is clear  that the set $\mathcal{P}_n$ of $n$-positive maps forms a closed convex cone, namely,
\begin{align}
\Phi \in \mathcal{P}_n &\Rightarrow c\Phi \in \mathcal{P}_{n},  \quad c \in \mathbb{R}^+,\label{cone_prop} \\
\Phi , \Psi \in \mathcal{P}_n &\Rightarrow\lambda \Phi + (1-\lambda) \Psi \in \mathcal{P}_{n} \;\;\forall \lambda \in (0,1), \label{convex_prop} \\
(\Phi_{n})_{n \in \mathbb{N}} \subset \mathcal{P}_{n} &\Rightarrow \lim_{n \rightarrow \infty} \Phi_{n} \in \mathcal{P}_{n}. \label{clos_prop}
\end{align}
In particular, we can spell out the properties of $n$-positivity and complete positivity in the classical case $\mathcal{A} = \mathbb{C}^d$: a positive matrix $A \in \mathcal{M}_{d}(\mathbb{R})$ is called $n$-positive if its $n$-fold ampliation 
\begin{equation}
 \label{ampl_map_cl}
A^{(n)}=  \mathrm{id}_n \otimes A : \mathcal{M}_{n}(\mathbb{C}) \otimes \mathbb{C}^d \rightarrow   \mathcal{M}_{n}(\mathbb{C}) \otimes \mathbb{C}^d,
\end{equation}
is positive. It is completely positive if this holds for any $n \geqslant 1$. Importantly, \textit{positivity is equivalent to complete positivity in the classical case}, indeed the amplified map~\eqref{ampl_map_cl} explicitly reads
\begin{equation}
     A^{(n)} =\mathrm{id}_n  \otimes A = (A_{ij})_{i,j=1}^n,\quad  A_{ij} :=  \begin{cases}
     A \;\;\;i = j\\
     0 \;\;\;\mbox{  otherwise}
     \end{cases},
\end{equation}
therefore $A^{(n)}$ is positive if and only if so is $A$. This result is a consequence of the commutativity of the $C^\ast$-algebra $(\mathbb{C}^d , \sbullet, \ast , \| \cdot \|)$, as discussed below. We refer
to~\cite{arveson1969,stinespring1955positive} for the proofs.
\begin{prop}
\label{comm_dom_range_pos_prop}
    Let $\mathcal{A},\mathcal{B}$ be $C^\ast$-algebras, and $\Phi : \mathcal{A} \mapsto \mathcal{B}$ be a positive map. Then, if either $\mathcal{A}$ or $\mathcal{B}$ is commutative, $\Phi$ is completely positive.
\end{prop}
In particular, the latter proposition allows us to conclude that any state $\varphi : \mathcal{A} \mapsto \mathbb{C}$ on a $C^\ast$-algebra $\mathcal{A}$ is automatically completely positive.   

Once the classical case has been understood , we can move to the quantum setting $\mathcal{A} = B(\mathcal{H})$. Physically, in the quantum realm, we can interpret any $n$-positive trace-preserving map $\Phi$ on $B(\mathcal{H})$ in the following way: the map $\Phi^{(n)} = \mathrm{id}_n \otimes \Phi$ describes in a legitimate way the evolution of a composite system $A + S$, in which the system $S$ evolves under $\Phi$, while the $n$-dimensional ancilla evolves trivially, but this is not generally true if we replace the ancilla $A$ with a $n^\prime$-dimensional ancilla $A^\prime$ where $n^\prime > n$. In order to clarify better the physical meaning of $n$-positivity,
let us recall two relevant results about $n$-positive and completely positive maps, first discussed in~\cite{Choi_1972}.
\begin{prop}
\label{n_pos_prop}
Let $d = \dim(\mathcal{H})$ and $\Phi$ be a positive map on $B(\mathcal{H})$. Then,
\begin{itemize}
\item[]$i.$ if $\Phi$ is $n$-positive, then it is $m$-positive for all $m < n$;
\item[]$ii.$ if $\Phi$ is $d$-positive, then it is completely positive.
\end{itemize}
\end{prop}
Physically speaking, property~\textit{i.} implies that, given any $n$-positive trace-preserving map $\Phi$, then the evolution $\mathrm{id}_{n^\prime} \otimes \Phi$ of the compound system $A+S$ is legitimate provided that $n^\prime \leqslant n$, while the implication~\textit{ii.} tells us that, given a $d$-dimensional system, if the map $\mathrm{id}_{d} \otimes \Phi $, with $\Phi$ being positive trace-preserving, maps states into states, then so does any map $\mathrm{id}_{n} \otimes \Phi$ with $n > d$. 

Now, as already noted in Subsection~\ref{open_dyn_sec}, the physical phenomenon behind $n$-positivity and, in particular, complete positivity is the occurrence of entangled states in bipartite quantum systems.       
Thus, in order to go deeper into this connection, let us recall several basic concepts of entanglement theory. Given two systems $A$ and $B$ with Hilbert spaces $\mathcal{H}_A$ and $\mathcal{H}_B$, any bipartite pure state $\ket{\psi_{AB}}$ of the composite system $A+B$ always admits a \textit{Schmidt decomposition}~\cite{schmidt1907theorie}
\begin{equation}
\ket{\psi_{AB}} = \sum_{i=1}^{r}s_i \ket{i_A} \otimes \ket{i_B},
\end{equation} 
where $\{ \ket{i_A} \}_{i=1}^{d_A}$ and $\{ \ket{i_B} \}_{i=1}^{d_B}$ are orthonormal bases of  $\mathcal{H}_{A}$ and $\mathcal{H}_{B}$, $s_i \geqslant 0$, and $r = \mbox{SR}(\ket{\psi_{AB}}) \leqslant  \min \{d_A , d_B \}$ is called the \textit{Schmidt rank} of $\ket{\psi_{AB}}$. If $\mbox{SR}(\ket{\psi_{AB}}) = 1$, then $\ket{\psi_{AB}} = \ket{\psi_{A}} \otimes \ket{\psi_{B}}$, i.e., the state is factorized, otherwise it is entangled. In particular, $\ket{\psi_{AB}}$ is said to be maximally entangled whenever $r = \min\{ d_A , d_B \}$ and $s_i = 1/\sqrt{r}$ for any $i=1, \dots , r$~\cite{heinosaari2011mathematical}. Therefore, we can regard the Schmidt rank as a quantifier of the entanglement content of a pure bipartite state.
The notion of Schmidt rank may be extended to bipartite density operators as follows~\cite{terhal_2000}
\begin{equation}
\label{Schmidt_num}
\mbox{SN}(\rho) := \min_{ \{ p_i , \ket{\psi_i} \} }\max_{i}\mbox{SR}(\ket{ \psi_i }),\quad \rho \in \mathcal{S}(\mathcal{H}_{A} \otimes \mathcal{H}_{B}),
\end{equation}  
where the minimum is taken over all the ensemble decompositions~\eqref{ens_dec} of $\rho$. The parameter $\mbox{SR}(\rho)$ is called the \textit{Schmidt number} of $\rho$ and, consistently, if $\rho = \ketbra{\psi}$, then $\mbox{SR}(\rho)=\mbox{SR}(\ket{\psi})$. 
The next theorem discusses a characterization of $n$-positive maps in terms of the Schmidt number~\cite{terhal_2000}.
\begin{thm}
Let $\dim(\mathcal{H}) = d$. A map $\Phi: B(\mathcal{H}) \mapsto B(\mathcal{H})$ is $n$-positive if and only if $\Phi^{(d)}(\rho)=  (\mathrm{id}_d \otimes \Phi )(\rho) \geqslant 0$
for all states $\rho\in \mathcal{S}(\mathcal{H}\otimes \mathcal{H})$ with Schmidt number $\mbox{SR}(\rho) \leqslant  n$, where $\mathrm{id}_d$ denotes the identity on $B(\mathcal{H})$.
\end{thm}
In particular, when $n=1$ and $n=d$, we find a result consistent with the definition of positivity and Proposition~\ref{n_pos_prop}. Therefore, in words, given a bipartite qudit system $A+B$, $n$-positive trace-preserving maps, with $n \leqslant d$, map all entangled states of A+B of Schmidt number at most equal to $n$ into states. Roughly speaking, $n$-positive trace-preserving maps, with $n=1,\dots , d-1$, describe in a physically consistent way the evolutions of states which are not too much entangled, where the amount of entanglement is measured by the Schmidt number of the state defined by~\eqref{Schmidt_num}.
The other way around, we can employ $n$-positive maps in order to detect bipartite entangled states with Schmidt number at least $n+1$~\cite{terhal_2000}:
\begin{thm}
\label{char_n_ent_states_thm}
    Let $\rho \in \mathcal{S}(\mathcal{H} \otimes \mathcal{H})$ be a state of a bipartite quantum system. Then $\rho$ is entangled with $\mbox{SR}(\rho) > n$ if and only if there exists a $n$-positive map $\Phi$ such that $(\mathrm{id}_d \otimes \Phi)(\rho) \not \geqslant 0$.   
\end{thm}
Thus, the contrapositive of Theorem~\ref{char_n_ent_states_thm} yields the following hierarchy of closed convex subsets of the set $\mathcal{S}(\mathcal{H} \otimes \mathcal{H})$ of states on $\mathcal{H} \otimes \mathcal{H}$
\begin{equation}
\label{hierarchy_set_states}
\mathrm{Sep}(\mathcal{H}) \equiv \mathrm{Sch}_1(\mathcal{H}) \subseteq \mathrm{Sch}_2(\mathcal{H})   \subseteq \dots \subseteq \mathrm{Sch}_d(\mathcal{H}) \equiv \mathcal{S}(\mathcal{H} \otimes \mathcal{H}),  \end{equation}
where $\mathrm{Sch}_n(\mathcal{H})  $ denotes the set of states on $\mathcal{H} \otimes \mathcal{H}$ with Schmidt number at most equal to $n$ and, in particular, $\mathrm{Sch}_1(\mathcal{H})   \equiv \mathrm{Sep}(\mathcal{H})$ stands for the set of separable states on $\mathcal{H} \otimes \mathcal{H}$, see Equation~\eqref{sep_def}.

Importantly, for bipartite qubit systems, we can replace the characterization of separability stated in Theorem~\ref{char_n_ent_states_thm} with the partial transposition criterion~\cite{horodecki1996necessary}, namely, it is sufficient (and necessary) to check the condition $(\mathrm{id}_n \otimes T)(\rho) \geqslant 0$ for separability of $\rho \in \mathcal{S}(\mathcal{H} \otimes \mathcal{H})$, where $T$ is the transposition map with respect to some basis of $B(\mathcal{H})$. The latter result is a consequence of the fact that, at variance with what happens at higher dimensions~\cite{choi1975positive,choi1980some}, all positive qubit maps are decomposable~\cite{stormer1963positive}. Recall that a positive map $\Phi : B(\mathbb{C}^d) \mapsto B(\mathbb{C}^d)$ is called \textit{decomposable} if it can be written as
\begin{equation}
    \Phi = \Phi_1 + T \circ \Phi_2, \quad \Phi_1,\Phi_2 \in \mathcal{CP}, 
\end{equation}
where $\mathcal{CP}$ denotes the cone of completely positive maps.
In particular, any map $\Phi$ on $B(\mathbb{C}^d)$ of the form $\Phi = T \circ \Phi^\prime$ for some $\Phi^\prime \in \mathcal{CP}$ is called \textit{completely copositive}~\cite{woronowicz1976positive}.

We can now enrich the hierarchy of intermediate notions between positivity and complete positivity by introducing the generalized Schwarz property~\cite{carlen2023characterizing,Lindblad76}. 
\begin{defn}[$n$-generalized Schwarz property]
    Let $\Phi$ be a positive map on $B(\mathcal{H})$. Then $\Phi$ is called generalized Schwarz if it satisfies the operator inequality
    \begin{equation}
\label{gen_OSI}
 \Phi(X^\ast X) \geqslant \Phi(X)^\ast \Phi(\mathbb{I})^{+} \Phi(X), \quad \forall X \in B(\mathcal{H}),
    \end{equation}
    where $Y^{+}$ denotes the Moore-Penrose pseudoinverse of $Y \in B(\mathcal{H})$~\cite{penrose1955generalized}.
    The map $\Phi$ is called $n$-generalized Schwarz if its $n$-fold ampliation $\Phi^{(n)}$ in~\eqref{ampl_map} is generalized Schwarz.
\end{defn}
In order to justify the name of the latter property, note that the operator inequality~\eqref{gen_OSI} is a special case of the inequality 
\begin{equation}
\label{op_ineq_2_pos}
    \Phi(X^\ast Y^{+} X) \geqslant \Phi(X)^\ast \Phi(Y)^{+} \Phi(X),
\end{equation}
holding for all $Y \geqslant 0$ and $X \in B(\mathcal{H})$ with $\Ker(Y) \subseteq \Ker(X^\ast)$, which, in turn, is an operator version of the Cauchy-Schwarz inequality~\cite{conway1994course}
\begin{equation}
\label{CS_states}
    |\varphi(X^\ast Y)|^2 \leqslant \varphi(X^\ast X) \varphi(Y^\ast Y) , \quad \forall X,Y \in B(\mathcal{H}),
\end{equation}
true for any state $\varphi : B(\mathcal{H}) \mapsto \mathbb{C}$.
If a generalized Schwarz map $\Phi$ is unital, then the operator inequality~\eqref{gen_OSI} reduces to 
\begin{equation}
\label{OSI}
    \Phi(X^\ast X) \geqslant \Phi(X)^\ast \Phi(X), \quad\forall X \in B(\mathcal{H}),
\end{equation}
usually called the operator Schwarz inequality and, accordingly, we will call $\Phi$ a \textit{Schwarz} map~\cite{wolf2012quantum}.  
Similarly, a unital $n$-generalized Schwarz map $\Phi$ will be simply called a $n$\textit{-Schwarz} map. 
Notice that the terminology is not standardized~\cite{fagnola1999quantum,siudzinska2021interpolating}: some authors call $\Phi$ a Schwarz map when it is positive and  satisfies~\eqref{gen_OSI} or the related inequality
\begin{equation}
\label{mod_Schw_ineq}
    \|\Phi(\mathbb{I}) \|\Phi(X^\ast X) \geqslant \Phi(X)^\ast \Phi(X), \quad \forall X \in B(\mathcal{H}),
\end{equation}
which also characterizes a property known as \textit{local complete positivity}, see~\cite{choi1974schwarz,stormer1963positive} for more details.
Observe also that, as it happens with $n$-positivity, the inequalities~\eqref{gen_OSI} and~\eqref{mod_Schw_ineq} do not imply \textit{any} condition on $\Phi(\mathbb{I})$, at variance with the operator Schwarz inequality~\eqref{OSI} that yields subunitality, namely, $\Phi(\mathbb{I}) \leqslant \mathbb{I}$.  

That being said, the physical meaning of the $n$-(generalized) Schwarz property is still not clear, at variance with the notion of $n$-positivity but, interestingly, this property often appears as a natural condition across different areas of mathematical physics. In particular, powerful results by Hiai and Petz~\cite{hiai2012quasi}, implying Lieb's concavity and convexity theorems~\cite{lieb1973convex} and the data processing inequality~\cite{Lindblad1975,muller2017monotonicity,uhlmann1977relative}, were known to be true for Schwarz maps~\cite{carlen2023characterizing}, which also emerge in recovery theorems~\cite{jenvcova2012reversibility}. Moreover, the operator Schwarz inequality~\eqref{OSI} turns out to be sufficient for the description of open-quantum-system asymptotics~\cite{amato2024decoherence,AFK_OSID,wolf2010inverse}. 

As for the sets $\mathcal{P}_n$, it is immediate to check that the sets $\mathcal{S}_n$ of $n$-generalized Schwarz maps are closed convex cones, viz. they satisfy properties~\eqref{cone_prop}-\eqref{clos_prop}.   
Moreover, it is clear from the definition that any $n$-generalized Schwarz map $\Phi$ is $n$-positive. 
Let us now discuss the analogue of Proposition~\ref{n_pos_prop} for $n$-generalized Schwarz maps. To this purpose, observe that, as noted in~\cite{carlen2023characterizing}, the $n$-generalized Schwarz property may be conveniently rephrased in matrix form as 
\begin{equation}
\label{n_gen_Schw_ineq_mat}
\begin{pmatrix}
    (\mathrm{id}_n \otimes \Phi)(\mathbb{I}_{nd}) & (\mathrm{id}_n \otimes\Phi)(X) \\
    (\mathrm{id}_n \otimes\Phi)(X)^\ast & (\mathrm{id}_n \otimes \Phi)(X^\ast X)
\end{pmatrix} \geqslant 0, \quad \forall X \in \mathcal{M}_{n}(\mathbb{C}) \otimes B(\mathcal{H}),
\end{equation}
where $\mathbb{I}_{\ell d}$ is the identity operator over $\mathbb{C}^\ell \otimes \mathcal{H}$,
as a consequence of the characterization of positive semi-definiteness based on Schur complements~\cite{zhang2006schur}.
\begin{prop}
\label{n_gen_Schw_propos}
Let $\Phi$ be a positive map on $B(\mathcal{H})$, with $\dim(\mathcal{H}) = d$. Then,
\begin{itemize}
\item[]$i.$ if $\Phi$ is a $(n+1)$-positive map, then it is a $n$-generalized Schwarz map;
    \item[]$ii.$ if $\Phi$ is $n$-generalized Schwarz, then $\Phi$ is $m$-generalized Schwarz for any $m < n$; 
    \item[]$iii.$ if $\Phi$ is $d$-generalized Schwarz, then $\Phi$ is completely positive.
\end{itemize}
\end{prop}
\begin{proof}
i. See~\cite{carlen2023characterizing}.

ii. Let us prove that, if $\Phi$ is $n$-generalized Schwarz, then $\Phi$ is $(n-k)$-generalized Schwarz for all $k = 1, \dots , n-1$. If $k \in \{ 1, \dots , n-1 \} $,
then the inequality~\eqref{n_gen_Schw_ineq_mat} for 
$X = \tilde{X} \oplus 0$, with $\tilde{X} \in \mathcal{M}_{n-k}(\mathbb{C}) \otimes B(\mathcal{H})$, reads
\begin{equation}
    \!\!\!\begin{pmatrix}
   (\mathrm{id}_{n-k} \otimes \Phi)(\mathbb{I}_{(n-k)d}) & &(\mathrm{id}_{n-k} \otimes \Phi)(\tilde{X}) &  \\
   & (\mathrm{id}_{k} \otimes \Phi)(\mathbb{I}_{kd}) & & 0_{kd}\\
    (\mathrm{id}_{n-k} \otimes \Phi)(\tilde{X})^\ast & &(\mathrm{id}_{n-k} \otimes \Phi) (\tilde{X}^{\ast} \tilde{X}) & \\
    & 0_{kd} & & 0_{kd}
\end{pmatrix} \geqslant 0,
\end{equation}
where $0_{\ell d}$ is the zero operator on $\mathbb{C}^\ell \otimes \mathcal{H}$ and the unspecified blocks are zero operators of suitable size. The latter condition implies that
\begin{equation}
    \begin{pmatrix}
   (\mathrm{id}_{n-k} \otimes \Phi)(\mathbb{I}_{(n-k)d}) & (\mathrm{id}_{n-k} \otimes \Phi)(\tilde{X})  \\
   (\mathrm{id}_{n-k} \otimes \Phi)(\tilde{X})^\ast &(\mathrm{id}_{n-k} \otimes \Phi) (\tilde{X}^{\ast} \tilde{X})
    \end{pmatrix} \geqslant 0,
\end{equation}
which is equivalent to $(n-k)$-generalized Schwarz property for the arbitrariness of $\tilde{X} \in \mathcal{M}_{n-k}(\mathbb{C}) \otimes B(\mathcal{H})$.

iii. If $\Phi$ is $d$-generalized Schwarz, then it is $d$-positive, therefore the result follows from claim~\textit{ii.} of Proposition~\ref{n_pos_prop}.
\end{proof}

In particular, the first implication of Proposition~\ref{n_gen_Schw_propos} implies that any $2$-positive map is generalized Schwarz and, in fact, it was proved in~\cite{carlen2023characterizing} that a linear map $\Phi$ is $2$-positive if and only if $\Phi(\mathbb{I}) \geqslant 0$ and the operator inequality~\eqref{op_ineq_2_pos} holds.
See again~\cite{carlen2023characterizing} for a similar characterization of the generalized Schwarz property.
On the basis of all the properties discussed so far, if $\mathcal{S}_n$ denotes the cone of $n$-generalized Schwarz maps, then we have 
\begin{equation}
\label{pos_prop_hierarchy}
\mathcal{CP} = \mathcal{P}_d = \mathcal{S}_d \subset \mathcal{P}_{d-1} \subset \dots \subset \mathcal{P}_2 \subset \mathcal{S}_1 \subset \mathcal{P}_1.
\end{equation}
Note that all the inclusions are strict. 

To illustrate this, consider the one-parameter family of unital trace-preserving positive maps 
\begin{equation}
\label{gen_red_map}
\Phi_a(X)= \frac{1}{d-a}(  \Tr(X)\mathbb{I} - aX), \quad X \in B(\mathcal{H}),\; a \in \mathbb{R}.
\end{equation}
Clearly, $\Phi_a$ is completely positive if $a < 0$.
Moreover, observe that, when $a = 1$, the map $\Phi_{a}$ reduces to the reduction map, appearing in entanglement theory~\cite{horodecki1999reduction}, while, when $a = \frac{1}{d-1}$, the map $\Phi_{a}$ reduces, up to a constant, to the example of $(d-1)$-positive but not completely positive map discussed by Choi in~\cite{Choi_1972}. Importantly, $n$-positivity and the $n$-Schwarz property can be characterized as follows~\cite{tomiyama1985geometry}:
\begin{align}
   \Phi_a &\in \mathcal{P}_n \Leftrightarrow a \leqslant \frac{1}{n},  \quad \quad 1\leqslant n \leqslant d, \\
   \Phi_a &\in \mathcal{S}_n \Leftrightarrow a \leqslant \frac{d}{1+nd} , \quad 1\leqslant n \leqslant d-1,
\end{align}
therefore, given $1 \leqslant n \leqslant d-1$, $\Phi_a$ is $n$-positive but not $n$-Schwarz when $ \frac{d}{1+nd} <  a \leqslant \frac{1}{n}$. See also~\cite{sun2022k} for related results. Furthermore, the strictness of the inclusions in the hierarchy~\eqref{pos_prop_hierarchy} implies by Theorem~\ref{char_n_ent_states_thm} the same property for the chain~\eqref{hierarchy_set_states}. 

Consider also the following example similar to~\eqref{gen_red_map}:
\begin{equation}
\label{gen_red_map_transp}
\Phi_{a,T}(X)= \frac{1}{d-a}(  \Tr(X)\mathbb{I} - aX^T), \quad X \in B(\mathbb{C}^d),\; a \in \mathbb{R}.
\end{equation}
When $a \geqslant 0$, one gets that~\cite{tomiyama1985geometry}
\begin{equation}
  \Phi_{a,T} \in \mathcal{P}_{1} \Leftrightarrow   \Phi_{a,T} \in \mathcal{CP} \Leftrightarrow a \in [0,1],
\end{equation}
namely, positivity and complete positivity for $\Phi_{a,T}$ with $a \geqslant 0$ are equivalent. Interestingly, if we relax the condition $a \geqslant 0$, then this is no longer true. Indeed,
\begin{align}
   \Phi_{a,T} \in \mathcal{P}_{1} &\Leftrightarrow a \leqslant 1, \quad    \Phi_{a,T} \in \mathcal{S}_{1} \Leftrightarrow \frac{2d}{1-\sqrt{4d+1}} \leqslant a \leqslant 1, \\
   \Phi_{a,T} \in \mathcal{P}_{n} &\Leftrightarrow   \Phi_{a,T} \in \mathcal{S}_{n} \Leftrightarrow -1 \leqslant a \leqslant 1, \quad 2 \leqslant n \leqslant d,
\end{align}
and thus positivity and Schwarz properties are not equivalent, at variance with $n$-positivity and $n$-Schwarz positivity for $n \geqslant 2$. In other words, by allowing the parameter $a$ to be negative, we break the equivalence between positivity and Schwarz property and between the latter property and $2$-positivity. In particular, if $a = -2$ and $d=2$, then we recover the map 
\begin{equation}
    \Phi_{-2,T}(X) = \frac{1}{2}\left( \Tr(X) \frac{\mathbb{I}}{2} + X^T \right), \quad X \in B(\mathbb{C}^2), 
\end{equation}
namely, the first example of Schwarz map not being $2$-positive shown in~\cite{choi1980some}. 

Besides their differences, both maps $\Phi_{a}$ and $\Phi_{a,T}$ are obtained by linearly combining a positive trace-preserving map with the contraction onto the maximally mixed state, viz. they are of the form
\begin{equation}
\label{SPA}
\Psi_{\lambda}(X) = \lambda \Tr(X) \frac{\mathbb{I}}{d} + (1-\lambda) \Phi(X), \quad X \in B(\mathcal{H}), \; \lambda \in \mathbb{R},
\end{equation}
where $\Phi = \mathrm{id}_d$ (respectively, $\Phi = T$) in the case of $\Phi_a$ (respectively, $\Phi_{a,T}$). When $\Phi \not \in \mathcal{CP}$ and the mixture $\Psi_{\lambda} \in \mathcal{CP}$, then this technique is known as \textit{structural physical approximation} of $\Phi$~\cite{bae2017designing,horodecki1999reduction,shultz2016structural}, and it is employed in entanglement theory in order to construct a physical approximation of a positive map that is able to detect some desired entangled state. More mathematically, the mixing operation~\eqref{SPA} can be employed to construct, from a positive map, new maps with a higher degree of positivity, such as Schwarz maps~\cite{chruscinski2020kadison}. 

Let us recall another intriguing notion lying between positivity and complete positivity. As anticipated in 
Subsection~\ref{open_dyn_sec}, in general the tensor product of two positive maps is not positive, while this is true if the two maps are either completely positive or, similarly, completely copositive. Motivated by this observation, we say that, given $n \in \mathbb{N}$, a positive map $\Phi$ is $n$\textit{-tensor-stable positive} if $\Phi^{\otimes n}$ is positive and \textit{tensor-stable positive} if it is $n$-tensor-stable positive for all $n \in \mathbb{N}$~\cite{filippov2017positive,hayashi2006quantum,muller2016positivity}. In particular, completely positive or completely co-positive maps are tensor-stable positive, and we can call any ($n$-)tensor-stable positive map \textit{non-trivial} if it is neither completely positive nor completely copositive. Importantly, for any $n \in \mathbb{N}$ and any $d \geqslant 2$, there exists a non-trivial $n$-tensor-stable map on $B(\mathcal{H})$ with $\dim(\mathcal{H}) = d$, while the non-existence of non-trivial tensor-stable maps was proved in the two-dimensional case~\cite{muller2016positivity}. 

To the best of our knowledge, the existence of non-trivial tensor-stable maps beyond the qubit case is still an open problem and, interestingly, it would imply the existence of non-positive partial transpose bound entangled states~\cite{bengtsson2017geometry}, which may be employed in quantum key distribution~\cite{horodecki2005secure,horodecki2008low} as well as for the violation of Bell inequalities~\cite{vertesi2014disproving} in spite of their undistillability.     

Finally, let us point out another possible way of interpolating positivity and complete positivity under the trace-preservation or unitality assumptions~\cite{siudzinska2021interpolating}. Specifically, the key observation is the fact that, for any Hermiticity-preserving trace-preserving map $\Phi$, we have~\cite{chruscinski2022dynamical} 
\begin{align}
\label{char_pos_Russo_Dye}
\Phi \in \mathcal{P} &\Leftrightarrow \| \Phi(X) \|_{1} \leqslant \| X \|_{1}, \quad \forall X^\ast = X \in B(\mathcal{H}), \\
\label{char_Compl_pos_Russo_Dye}
\Phi \in \mathcal{CP} &\Leftrightarrow \| (\mathrm{id}_d \otimes \Phi)(X) \|_{1} \leqslant \| X \|_{1}, \quad \forall X^\ast = X \in \mathcal{M}_{d}(\mathbb{C}) \otimes B(\mathcal{H}),
\end{align}
where $\|X\|_{1} := \Tr(|X|)$ denotes the trace norm of $X \in B(\mathcal{H})$. 
Let us denote with $ \mathcal{V}(S) := \mbox{span}_{\mathbb{R}}(S)$ the real linear subspace generated by a finite subset $S \subset B(\mathcal{H})$. Then we say that a Hermiticity-preserving trace-preserving map $\Phi$ is $n$-\textit{partially contractive} if and only if, for any linearly independent set of states $\{ \rho_{1} , \dots , \rho_{n} \}$, one has 
\begin{equation}
    \| (\mathrm{id}_d \otimes \Phi)(X) \|_{1} \leqslant \| X \|_1, \quad \forall X \in  B(\mathcal{H})_{\mathrm{h}} \otimes  \mathcal{V}(\{ \rho_{1}, \dots , \rho_{n} \}),
\end{equation}
where $B(\mathcal{H})_{\mathrm{h}}$ is the space of Hermitian operators on $\mathcal{H}$.
Thus, as a consequence of~\eqref{char_pos_Russo_Dye} and~\eqref{char_Compl_pos_Russo_Dye}, we have the following hierarchy of properties for positive trace-preserving maps
\begin{equation}
\mathcal{CPTP} = \mathcal{C}_{d^2} \subseteq \mathcal{C}_{d^2 - 1} \subseteq \dots \subseteq \mathcal{C}_{2} \subseteq \mathcal{C}_{1} = \mathcal{PTP},  
\end{equation}
where $\mathcal{C}_{n}$ denotes the convex set of $n$-partially contractive maps, and $\mathcal{PTP}$ ($\mathcal{CPTP}$) denotes the set of (completely) positive trace-preserving maps. Of course, an analogous hierarchy of properties may be introduced for unital positive maps in view of their contractivity with respect to the operator norm~\cite{chruscinski2022dynamical}.

\section{Classical-quantum correspondence and back}
\label{class_quant_sec}
In this section we  address another aspect of the interplay between the classical and quantum realms: how classical states and evolutions can be mapped into quantum ones, and vice versa. 
From now on, let us fix an orthonormal basis $\mathcal{B}=\{ \ket{i} \}_{i=1}^d$ of $\mathcal{H}$, with respect to which any operator on $\mathcal{H}$ can be represented as a matrix of order $d$, and upon which any mapping between classical and quantum states (or evolutions) will depend.

As already said, classical states are described by probability vectors $\vec{p} \in \Delta_d$, while quantum states are given by density operators $\rho \in \mathcal{S}(\mathcal{H})$, therefore we are now going to revisit several mappings between $\Delta_d$ and $\mathcal{S}(\mathcal{H})$.  
The first natural correspondence between classical and quantum states which we consider is the following
\begin{equation}
\label{emb_1}
\begin{split}
\Gamma: \Delta_d &\rightarrow \mathcal{S}(\mathcal{H}) , \\
 \vec{p} = \begin{pmatrix}
p_1 \\
p_2 \\
\vdots \\
p_d
\end{pmatrix} \mapsto \rho_{\vec{p}} \equiv \Gamma(\vec{p}) &:= \begin{pmatrix}
p_1 & & & \\
& p_2 & & \\
& & \ddots & \\
& & & p_d
\end{pmatrix}
= \sum_{i=1}^d p_i \ketbra{i}{i},
\end{split}
\end{equation}
namely, we regard $\vec{p}$ as the diagonal of the corresponding density operator $\rho_{\vec{p}}$ or, in other words, the populations of $\rho_{\vec{p}}$ are the components of $\vec{p}$. An alternative mapping between $\Delta_d$ and $\mathcal{S}(\mathcal{H})$, introduced in the literature since the 1950s~\cite{riccia1966wave,LOINGER1962,schonberg1952application,sudarshan1976interaction}, reads
\begin{equation}
\label{classic_wave_func_map}
\Pi_{\vec{\phi}}: \Delta_d \rightarrow \mathcal{S}(\mathcal{H}), \qquad  \vec{p} = \begin{pmatrix}
p_1 \\
p_2 \\
\vdots \\
p_d
\end{pmatrix} \mapsto \Pi_{\vec{\phi}}(\vec{p}) := \ketbra*{\psi_{\vec{p}}^{(\vec{\phi})}}{\psi_{\vec{p}}^{(\vec{\phi})}}, 
\end{equation}
where
\begin{equation}
\label{class_wave_fun}
    \ket*{\psi_{\vec{p}}^{(\vec{\phi})}} := \begin{pmatrix}
\sqrt{p_1}e^{\mathrm{i} \phi_1} \\
\sqrt{p_2}e^{\mathrm{i} \phi_2} \\
\vdots \\
\sqrt{p_d}e^{\mathrm{i} \phi_d}
\end{pmatrix} = \sum_{j=1}^d \sqrt{p_j} e^{\mathrm{i} \phi_j} \ket{j},
\end{equation}
 with $\vec{\phi}=(\phi_{i})_{i=1}^d \in \mathbb{R}^d$. In this case, at variance with  $\Gamma$, the map $\Pi_{\vec{\phi}}$ yields density operators with nonzero coherences, namely, off-diagonal elements, dependent on $\vec{\phi}$. Moreover, we can interpret $\ket*{\psi_{\vec{p}}^{(\vec{\phi})}}$ as a classical ``wave function'' on a finite-dimensional Hilbert space associated with the probability distribution $\vec{p}$ via Born's rule, where the $\phi_i$ are \textit{classically} unobservable phase factors. Consistently, it is possible to derive the Fisher classical information metric on $\Omega = \{ 1, \dots , d\}$ from the Fubini-Study one on the set of pure states by using the mapping $\Pi_0$, while the general correspondence $\Pi_{\vec{\phi}}$ allows us to recover the Fisher quantum information metric defined in the standard sense~\cite{helstrom1967minimum}, see~\cite{facchi2010classical} for more details.
 
 Concerning dynamics, we have already noted in Subsection~\ref{classic_sec} that the classical evolution of probability vectors is described by a column-stochastic matrix, whereas the Schr{\"o}dinger dynamics, governing the time evolution of quantum states, is given by a quantum channel (see Subsection~\ref{open_dyn_sec}). Thus, we shall now describe a possible mapping between column-stochastic matrices and quantum channels (with respect to the incoherent basis $\mathcal{B}$ of $\mathcal{H}$).
 
Specifically, given a column-stochastic matrix $S = (S_{ij})_{i,j=1}^d$, we can associate with it a quantum channel $\Phi_S$ in the following natural way:
\begin{equation}
\label{class_channel}
S \xmapsto{\Phi} \Phi_S, \quad \Phi_S(X) := \sum_{i,j=1}^d S_{ij} \langle j |X| j \rangle \ketbra{i}{i}, \;\; X \in B(\mathcal{H}).
\end{equation}
Complete positivity of $\Phi_{S}$ immediately follows from Proposition~\ref{comm_dom_range_pos_prop}, while trace-preservation is a consequence of property~\eqref{norm_col-stoch}.
Notice that the channel $\Phi_S$ defined by~\eqref{class_channel} corresponds to the column-stochastic matrix $S$ in the following sense 
\begin{equation}
\label{class_quant_corresp}
S_{ij} = \Tr(\ketbra{i}{i}\Phi_S(\ketbra{j}{j})), \quad i,j=1, \dots , d,
\end{equation}
implying that, by adopting the embedding~\eqref{emb_1}, 
\begin{equation}
\label{quantum_rep}
\vec{p}^{\,\prime} = S \vec{p} \iff \rho_{\vec{p}^{\,\prime}} = \Phi_S(\rho_{\vec{p}}).
\end{equation}
 It is also easy to check that the mapping $\Phi$ defined by~\eqref{class_channel} is a representation of the semigroup of column-stochastic matrices on $B(\mathcal{H})$, namely, $\Phi_{TS} = \Phi_T \circ \Phi_S$ for any pair $S,T$ of column-stochastic matrices.

 Therefore, the embedding $\Gamma$ is equivariant under the action of the Markov semigroup of stochastic matrices, namely $\Gamma \circ S = \Phi_S\circ \Gamma$, which can be understood as the  commutative diagram
\begin{equation}
  \begin{tikzcd}
\Delta_d \arrow[r,"S"] \arrow[d,"\Gamma"'] & \Delta_d \arrow[d,"\Gamma"] \\
\mathcal{S}(\mathcal{H}) \arrow[r,"\Phi_S"'] & \mathcal{S}(\mathcal{H})
\end{tikzcd}  
\end{equation}
which holds for all column-stochastic matrices $S$.
By using the mapping~\eqref{class_channel}, we can thus associate  a discrete-time (respectively, continuous-time) semigroup of column-stochastic matrices with a discrete-time (respectively, continuous-time) semigroup of quantum channels. 

Moreover, in the continuous-time case, we can obtain a quantum embedding of a classical Markov evolution in a different way: given a classical dynamics $(e^{tL})_{t\in\mathbb{R}^+}$ with Kolmogorov generator $L = (L_{ij})_{i,j=1}^d$ of the form~\eqref{Kolmogorv_struc}, we can introduce the following GKLS generator 
\begin{equation}
\label{class_gen}
\begin{split}
    \mathcal{L}_{L}(X) &:= \sum_{i=1}^d \left( \sum_{j=1}^d W_{ij}  \langle j |X| j \rangle  \ketbra{i}{i} - \frac{1}{2} W_{i} \{ \ketbra{i} , X \}  \right) \\&= \Phi_{W}(X) - \frac{1}{2} \{ \Phi_{W}^\dagger(\mathbb{I}) , X \}, \quad X \in B(\mathcal{H}),
    \end{split}
\end{equation}
with $\Phi_{W}$ being the completely positive map corresponding to the non-negative matrix $W = (W_{ij})_{i,j=1}^d$ via~\eqref{class_channel}, where we set $W_{ii} :=0$ for all $i=1, \dots , d$. 
Clearly, the generator~\eqref{class_gen} fulfills the condition
\begin{equation}
\label{class_quant_corresp_gen}
L_{ij} = \Tr(\ketbra{i}{i}\mathcal{L}_{L}(\ketbra{j}{j})), \quad i,j=1, \dots , d,
\end{equation}
which is the analogue of~\eqref{class_quant_corresp}. Interestingly, the generator~\eqref{class_gen} is dissipative, i.e., it has a GKLS representation with zero Hamiltonian part. In this sense and in our finite-dimensional setting, we can state that there is no classical analogue of a Hamiltonian GKLS generator, i.e., $\mathcal{L}(X) = -\mathrm{i} [H,X]$ with $H=H^\ast$ and $X \in B(\mathcal{H})$.

Let us now discuss how to reduce quantum states and evolutions into classical ones. First, a quantum state $\rho$ can be associated with a family of classical probability vectors. Given any  orthonormal basis $\mathcal{B}=\{\ket{i}\}_{i=1}^d$, we define the corresponding probability vector
\begin{equation}
\label{prob_vec_rho}
  \rho \mapsto  \Omega_{\mathcal{B}}(\rho) = \vec{p} = (p_i)_{i=1}^d, \quad p_i := \langle i|\rho| i \rangle , \quad i=1, \dots , d .
\end{equation}
Thus each orthonormal basis $\mathcal{B}$ provides a probability vector $\vec{p}$.

In this sense, a single quantum state $\rho$ determines an infinite family of such probability vectors---one for every possible basis.
However, as pointed out in~\cite{man2008semigroup}, knowing the full family $\{\Omega_{\mathcal{B}}(\rho)\}_{\mathcal{B}}$ is more than enough: it suffices to know a finite ``quorum'' of them---called quantum tomograms---to reconstruct 
$\rho$ completely. 

Now let us overview classical reductions of a given quantum dynamics. 
Just as a quantum state $\rho$ can be projected onto classical probability vectors, a quantum channel or, more generally, a positive trace-preserving map $\Phi$ can be reduced to a family of stochastic matrices acting on these classical shadows.
Given an orthonormal basis, the associated  column-stochastic matrix reads
\begin{equation}
\label{quant_class_corresp}
S_{ij} = \Tr(\ketbra{i}{i}\Phi(\ketbra{j}{j})), \quad i,j=1, \dots , d.
\end{equation}
It is worth observing that, given a linear map $\Phi$, the fact that the matrix $S = (S_{ij})_{i,j=1}^d$ defined by~\eqref{quant_class_corresp} is stochastic for any choice of the basis $\{ \ket{i} \}_{i=1}^d$ of $\mathcal{H}$ does \textit{not} guarantee that $\Phi$ is positive and trace-preserving and, in particular, a quantum channel, see~\cite{benatti2024quantum} for a counterexample. In fact, one has to require the stronger condition
\begin{equation}
    \langle \psi | \Phi(\ketbra{\varphi}{\varphi}) | \psi \rangle \geqslant 0, \quad  \forall \ket{\psi},\ket{\varphi} \in \mathcal{H},
\end{equation}
also involving non-orthogonal pairs of vectors.

Let us discuss more in depth the mapping~\eqref{quant_class_corresp} in the case of unitary evolutions. Let 
\begin{equation}
 \Phi_U(X) = UXU^\ast,    
\end{equation}
with $U$ being unitary and $X \in B(\mathcal{H})$, and compute
\begin{equation}
    \label{S_unit_case}
S_{ij} = \langle i |\Phi_U(\ketbra{j}{j})| i \rangle = \langle i | U \ketbra{j}{j} U^\ast | i \rangle = |U_{ij}|^2, \quad   U_{ij} := \langle i | U | j \rangle .
\end{equation}
Thus $S$ is orthostochastic and, consequently, column-stochastic.
A less direct but instructive derivation of~\eqref{S_unit_case}, which is more similar to what is done in~\cite{man2008semigroup}, is the following. The action of $\Phi_U$
on the classical state $\rho_{\vec{p}}$, corresponding to a probability vector $\vec{p}$, reads
\begin{equation}
\rho_{\vec{p}} = \sum_{i=1}^d p_i \ketbra{i}{i}
\mapsto
\Phi_U(\rho_{\vec{p}}) = \sum_{i,j=1}^d \left( \sum_{k=1}^d U_{ik} \overline{U}_{jk} p_k \right) \ketbra{i}{j}, 
\end{equation}
where the off-diagonal elements of the output matrix $\Phi_{U}(\rho_{\vec{p}})$ with respect to the basis $\{ \ket{i} \}_{i=1}^d$ are generally non-zero,
and it again induces the transformation on $\Delta_d$
\begin{equation}
\vec{p} = (p_i)_{i=1}^d \mapsto \vec{p}^{\,\prime} = S \vec{p}, \qquad p_i^\prime := \sum_{k=1}^d |U_{ik}|^2 p_k, \; i=1, \dots , d,
\end{equation}
where $S$ is again defined by~\eqref{S_unit_case}.
Then we can summarize the latter argument in this way
\begin{equation}
 \vec{p} \xmapsto{\Gamma} \rho_{\vec{p}} \xmapsto{\Phi_U} \Phi_U(\rho_{\vec{p}}) \xmapsto{\Omega} \vec{p}^{\,\prime} = S \vec{p} ,  
\end{equation}
for any $\vec{p} \in \Delta_d$. The above discussion of the mapping~\eqref{quant_class_corresp} may be generalized in a straightforward way to any quantum channel $\Phi$ by using the Kraus-Sudarshan representation~\cite{Kraus_71,sudarshan1961stochastic}. 
 
Up to now, we have seen how a column-stochastic matrix can be constructed from  a quantum channel  via the mapping~\eqref{quant_class_corresp}. Correspondingly, a discrete-time (or continuous-time)  semigroup of quantum channels naturally induces a corresponding discrete-time (or continuous-time) semigroup of column-stochastic matrices. 

Moreover, in the continuous-time case, we can obtain a classical Markov evolution from a quantum one by directly reducing the GKLS generator, as follows. Given a quantum dynamical semigroup $(\Phi_t)_{t\in\mathbb{R}^+}$ with generator $\mathcal{L}$, then we can construct the following Kolmogorov generator 
\begin{equation}
\label{quant_class_corresp_gen}
L_{ij} := \Tr(\ketbra{i}{i}\mathcal{L}(\ketbra{j}{j})), \quad i,j=1, \dots , d,
\end{equation}
which can be employed for the characterization of GKLS generators. 
In fact, according to a seminal result~\cite{kossa1972} by Kossakowski, $\mathcal{L}$ is the generator of a semigroup of positive trace-preserving maps if and only if, for \textit{any} choice of the basis $\{ \ket{i} \}_{i=1}^d$ of $\mathcal{H}$, the matrix $L = (L_{ij})_{i,j=1}^d$ with $L_{ij}$ given by~\eqref{quant_class_corresp_gen} is the generator of a semigroup of column-stochastic matrices. Interestingly, the latter characterization is the key finding for the derivation of the structure of quantum Markov generators presented in~\cite{GKS76}. Notice also that the classical reduction of the GKLS generator $\mathcal{L}$ is \textit{not} equivalent to reducing each map of the semigroup $(\Phi_t)_{t\in\mathbb{R}^+}$ according to~\eqref{quant_class_corresp}, as discussed in~\cite{benatti2024quantum}. 


We have already pointed out in Section~\ref{Markov_sec}  that column-stochastic matrices as well as quantum channels form a semigroup. Recall also that these two sets of maps are not groups, in fact not all column-stochastic matrices (respectively, quantum channels) are invertible and, moreover, the inverse of a column-stochastic matrix (respectively, a quantum channel) is column-stochastic (respectively, a quantum channel) if and only if it is a permutation matrix (respectively, a unitary)~\cite{buscemi2005clean,varga_book}.

However, we shall now see that it is possible to embed the two semigroups into groups~\cite{man2008semigroup}. If we start with the set of column-stochastic matrices, we can consider \textit{pseudo-stochastic} matrices~\cite{chruscinski2015pseudo} defined by the conditions
\begin{equation}
    S_{ij} \in \mathbb{R},  \quad \sum_{i=1}^d S_{ij} = 1, \quad i,j=1, \dots , d.
\end{equation}
In particular, the set of invertible pseudo-stochastic matrices is a group, containing the set of invertible column-stochastic matrices and, indeed, it is a Lie group isomorphic to the inhomogeneous  real general linear group $\mbox{IGL}(d-1,\mathbb{R})$ of order $d-1$~\cite{man2008semigroup}. Analogously, we can introduce the set of \textit{pseudo-bistochastic} matrices defined as follows 
\begin{equation}
    S_{ij} \in \mathbb{R},  \quad \sum_{i=1}^d S_{ij} = \sum_{j=1}^d S_{ij} =  1, \quad i,j=1, \dots , d,
\end{equation}
whose subset of invertible ones is a Lie group isomorphic to the homogeneous real general linear group $\mbox{GL}(d-1,\mathbb{R})$ of order $d-1$~\cite{man2008semigroup}. 

Regarding quantum channels, the set of invertible Hermiticity-preserving trace-preserving maps forms a group, which contains the group  of invertible quantum channels. Similarly, the set of invertible unital Hermiticity-preserving trace-preserving maps is a group containing the invertible unital quantum channels. As we previously observed,  the embedding $\Phi$ defined in~\eqref{quant_class_corresp} provides a representation of column-stochastic matrices on $B(\mathcal{H})$, thereby preserving their semigroup structure. In contrast, $\Phi_S$ is not invertible for any column-stochastic matrix $S$, and thus it \textit{cannot} respect the group structure of the set of invertible pseudo-stochastic matrices.  
\section{Conclusions}
\label{concl_sec}
In this work, we discussed various aspects of the connection between classical and quantum frameworks.
In particular, we provided an updated overview of the properties that lie between positivity and complete positivity,  introduced in the literature to shed light on the structure of the set of positive maps---a structure that is highly non-trivial in the noncommutative case. As discussed in detail in Sections~\ref{Markov_sec} and~\ref{pos_compl_pos_sec}, beyond their mathematical significance in operator algebra theory, many of these positivity properties naturally arise in contexts such as quantum dynamics and entanglement theory.
We further illustrated, in Section~\ref{class_quant_sec}, how to map quantum states and evolutions to classical ones and vice versa, which can be framed within the broader program of applying classical techniques in the quantum realm and quantum techniques in the classical realm, as outlined in Section~\ref{intro_sec}. 
We hope that this work will serve as a useful resource for both the mathematical and physical communities interested in the interplay between classical and quantum theories.

	\section*{Acknowledgements}
DA and PF acknowledge support from INFN through the project ``QUANTUM'', from the Italian National Group of Mathematical Physics (GNFM-INdAM), and from the Italian funding within the ``Budget MUR - Dipartimenti di Eccellenza 2023--2027''  - Quantum Sensing and Modelling for One-Health (QuaSiModO). PF acknowledges support from Italian PNRR MUR project CN00000013 -``Italian National Centre on HPC, Big Data and Quantum Computing''. DA acknowledges support from PNRR MUR project PE0000023-NQSTI.

\endpaper
\end{document}